\newcommand{\gwsimple}{\ensuremath{O( \log |X|) }}
\newcommand{\gwbadapprox}{\ensuremath{ \Omega(m / \log m) }}
\newcommand{\mfdpos}{MFD$_{\mathbb{N}}$\xspace}
\newcommand{\mfdint}{MFD$_{\mathbb{Z}}$\xspace}
\newcommand{\mfdy}{MFD$_{\mathbb{Y}}$\xspace}
\newcommand{\mccdint}{MCCD$_{\mathbb{Z}}$\xspace}
\newcommand{\mccdpos}{MCCD$_{\mathbb{N}}$\xspace}
\newcommand{\mccdy}{MCCD$_{\mathbb{Y}}$\xspace}
\newcommand{\mccdsizey}[2]{\textsf{mccd}_{\mathbb{Y}}(#1,#2)}
\newcommand{\mccdsizeint}[2]{\textsf{mccd}_{\mathbb{Z}}(#1,#2)}

\newcommand{\mfdsizeint}[2]{\textsf{mfd}_{\mathbb{Z}}(#1,#2)}
\newcommand{\mfdsizepos}[2]{\textsf{mfd}_{\mathbb{N}}(#1,#2)}
\newcommand{\mfdsizey}[2]{\textsf{mfd}_{\mathbb{Y}}(#1,#2)}
\newcommand{\width}[1]{\textsf{width}(#1)}
\newcommand{\mcfc}[1]{\textsf{mccc}(#1)}
\newcommand{\mcfcS}[1]{\ensuremath\textsf{mccc}_S(#1)}
\newcommand{\mcfcruntime}{\ensuremath{O(n\log m(m+n\log n))}}
\newcommand{\kfwapos}{$k$-FWA$_{\mathbb{N}}$\xspace}
\newcommand{\kfwaint}{$k$-FWA$_{\mathbb{Z}}$\xspace}
\newcommand{\flowcost}[1]{\Tilde{c}(#1)}

\documentclass[a4paper,UKenglish,cleveref, autoref, thm-restate]{lipics-v2021}

\theoremstyle{definition}

\Crefname{property}{Property}{Properties}

\usepackage[ruled]{algorithm}
\usepackage{algorithmic}
\usepackage{xspace}
\usepackage{mathtools}

\title{Width Helps and Hinders Splitting Flows}

\author{Manuel C\'{a}ceres}{Department of Computer Science, University of Helsinki, Finland}{manuel.caceresreyes@helsinki.fi}{https://orcid.org/0000-0003-0235-6951}{}

\author{Massimo Cairo}{Department of Computer Science, University of Helsinki, Finland}{cairomassimo@gmail.com}{}{}

\author{Andreas Grigorjew}{Department of Computer Science, University of Helsinki, Finland}{andreas.grigorjew@helsinki.fi}{https://orcid.org/0000-0003-0989-2415}{}

\author{Shahbaz Khan}{Department of Computer Science and Engineering, Indian Institute of Technology Roorkee, India}{shahbaz.khan@cs.iitr.ac.in}{ https://orcid.org/0000-0001-9352-0088}{}

\author{Brendan Mumey}{School of Computing, Montana State University, United States}{brendan.mumey@montana.edu}{https://orcid.org/0000-0001-7151-2124}{}

\author{Romeo Rizzi}{Department of Computer Science,  University of Verona, Italy}{romeo.rizzi@univr.it}{https://orcid.org/0000-0002-2387-0952}{}

\author{Alexandru I. Tomescu}{Department of Computer Science, University of Helsinki, Finland}{alexandru.tomescu@helsinki.fi}{https://orcid.org/0000-0002-5747-8350}{}

\author{Lucia Williams}{School of Computing, Montana State University, United States}{lgw2@uw.edu}{https://orcid.org/0000-0003-3785-0247}{}

\authorrunning{M. C\'{a}ceres et al.} 

\Copyright{Manuel C\'{a}ceres, Massimo Cairo, Andreas Grigorjew, Shahbaz Khan, Brendan Mumey, Romeo Rizzi, Alexandru Tomescu, and Lucia Williams} 



\category{} 

\relatedversion{} 

\funding{This work was partially funded by the European Research Council (ERC) under the European Union's Horizon 2020 research and innovation programme (grant agreement No.~851093, SAFEBIO), partially by the Academy of Finland (grants No.~322595, 328877), and partially by the National Science Foundation (NSF) (grants No.~1661530,1759522).}

\nolinenumbers 

\EventEditors{John Q. Open and Joan R. Access}
\EventNoEds{2}
\EventLongTitle{42nd Conference on Very Important Topics (CVIT 2016)}
\EventShortTitle{CVIT 2016}
\EventAcronym{CVIT}
\EventYear{2016}
\EventDate{December 24--27, 2016}
\EventLocation{Little Whinging, United Kingdom}
\EventLogo{}
\SeriesVolume{42}
\ArticleNo{23}

\begin{document}

\maketitle

\begin{abstract}
Minimum flow decomposition (MFD) is the NP-hard problem of finding a smallest decomposition of a network flow/circulation $X$ on a directed graph $G$ into weighted source-to-sink paths whose superposition equals $X$. 
We show that, for acyclic graphs, considering the \emph{width} of the graph
(the minimum number of paths needed to cover all of its edges)
yields advances in our understanding of its approximability.
For the version of the problem that uses only non-negative weights,
we identify and characterise a new class of \emph{width-stable} graphs,
for which a popular heuristic is a \gwsimple-approximation ($|X|$ being the total flow of $X$),
and strengthen its worst-case approximation ratio from $\Omega(\sqrt{m})$ to $\Omega(m / \log m)$ for sparse graphs, where $m$
is the number of edges in the graph.
We also study a new problem on graphs with cycles, Minimum Cost Circulation Decomposition (MCCD), and show that it generalises MFD through a simple reduction. For the version allowing also negative weights, we give a $(\lceil \log \Vert X \Vert \rceil +1)$-approximation ($\Vert X \Vert$
being the maximum absolute value of $X$ on any edge) using a power-of-two approach, combined with parity fixing arguments and a decomposition of unitary circulations ($\Vert X \Vert \leq 1$), using a generalised notion of width for this problem.
Finally, we disprove a
conjecture about the linear independence of minimum (non-negative) flow decompositions posed by Kloster et al.
\cite{kloster2018practical}, but show that its useful implication (polynomial-time assignments of weights to a given
set of paths to decompose a flow) holds for the negative version.
\end{abstract}



   

\section{Introduction}

Minimum flow decomposition (MFD) is the problem of finding a smallest sized decomposition 
of a network flow $X$ on  directed graph $G=(V,E)$ into weighted source-to-sink
paths whose superposition equals $X$.
We focus on the case where path weights are restricted to be integers.
It is a textbook result~\cite{ahujia1993network} that if $G$ is acyclic (a DAG) a decomposition using no more than $m = |E|$ paths
always exists.
However, MFD is strongly NP-hard~\cite{vatinlen2008simple}, even on DAGs, and even when the flow values come
only from $\{1, 2, 4\}$ \cite{hartman2012split}. 
Recent work has shown that the problem is FPT in the size of the
minimum decomposition~\cite{kloster2018practical}
and that it can be formulated as an ILP of quadratic size~\cite{Dias:2022uv}.

While difficult to solve, MFD is a key step in many applications. For example, MFD on DAGs is used to reconstruct
biological sequences such as RNA transcripts~\cite{pertea2015stringtie,tomescu2013novel,gatter2019ryuto,bernard2014efficient,tomescu2015explaining,williams2019rna,dias2023safety} and
viral strains~\cite{baaijens2020strain}. MFD can also be used to model problems in
networking~\cite{vatinlen2008simple, hartman2012split,mumey2015parity} and transportation planning~\cite{Olsen:2020aa},
although in some of these applications there may be cycles in the input.
Despite the ubiquity of the MFD problem, the gap in our knowledge about the approximability of MFD is large.
It is known~\cite{hartman2012split} that MFD (even on DAGs) is APX-hard 
(i.e., there is some $\epsilon > 0$ such that it is NP-hard to
approximate within a $(1+\epsilon)$ factor), so in particular,
MFD does not admit a PTAS, unless P=NP.
Furthermore, the best known approximation ratio is $\lambda^{\log \Vert X \Vert} \log \Vert X \Vert$~\cite{mumey2015parity},
where $\lambda$ is the length of the longest source-to-sink path and $\Vert X \Vert$ is the
largest flow value in the network.
In this work, we attempt to fill in some of
the gaps 
between these results. 


A natural lower bound for the size of an MFD of a DAG is the size of a \emph{minimum path cover} of the set of edges with non-zero flow (i.e.,~the minimum number of paths such that every such edge appears in \emph{at least} one path)---this size is called the \emph{width} of the network. This trivially holds because every flow decomposition is also such a path cover. These two notions are analogies of the more standard notions of path cover and width of the \emph{node set}.
The node-variants are classical concepts, with algorithmic results dating back to Dilworth and Fulkerson~\cite{dilworth2009decomposition,fulkerson1956note}. Despite this, the width has not been given any attention in the MFD problem, and in particular it has never been used in approximation algorithms to our knowledge. In this paper, we show that the width can play a key role both in the analysis of popular heuristics, and in obtaining the first approximation algorithm for a natural generalisation of MFD.


We start by considering the connections between the width and a popular 
heuristic algorithm for \mfdpos which we call \emph{greedy-weight}\footnote{Previous work has consistently
referred to this algorithm as \emph{greedy-width}. To avoid confusion with the width of the graph,
we introduce the name greedy-weight in this work.}~\cite{vatinlen2008simple},
which builds a flow decomposition by successively choosing the path that
can carry the largest flow.
Greedy-weight is commonly used in applications
(see e.g., \cite{tomescu2013novel,baaijens2020strain,pertea2015stringtie} among many), and it
seems to be mentioned in nearly every publication addressing flow decomposition. 
%
First, on sparse graphs we improve (i.e.,~increase) the worst-case lower bound for the greedy-weight approximation factor from $\Omega(\sqrt{m})$~\cite{hartman2012split}, showing for the first time that greedy-weight can be exponentially worse than the optimum:
\begin{restatable}{theorem}{ratiomfdpos}
\label{thm:gwbad}
The approximation ratio for greedy-weight on \mfdpos~is \gwbadapprox~for sparse graphs, in the worst case.
\end{restatable}
For this we use a class of sparse graphs where the optimum flow decomposition has size $O(\log m)$ whereas the greedy-weight algorithm returns
a solution of size $\Omega(m)$, only a constant factor away from the trivial decomposition. 
The key to this new bound is to design an input where the width
increases exponentially when a path is greedily removed. We also show that the same bound also holds for other greedy heuristics choosing instead the longest or shortest paths. 
%
Second, we identify a new class of graphs, defined by the property that their width does not increase as source-to-sink paths are removed (see~\Cref{prop:width} of width-stable graphs). We show a relation of width-stable graphs to funnels: precisely, a graph is not width-stable if it contains a funnel subgraph and a certain \textit{central} path. This is precisely the structure of the class of sparse graphs improving the approximation ratio of greedy-weight in~\Cref{thm:gwbad}. We also show that width-stability enables greedy-weight to remove paths of large enough flow~(\Cref{lem:enough-flow}), leading to the following result, with $|X|$ being equal to the total flow of the graph: 
\begin{restatable}{theorem}{gwapprox}
\label{thm:gwapproximation}
Let $G = (V, E)$ be a width-stable graph and $X: E \to \mathbb{N}$ a flow.
Greedy-weight is a \gwsimple -approximation for \mfdpos on $(G,X)$.
\end{restatable}

A notable example of width-stable graphs is the class of \emph{series-parallel graphs}; see~\cite{EPPSTEIN199241,tarjan-series-parallel} for fast recognition algorithms and pointers to other NP-hard problems that are easier on this class of graphs.
Series-parallel graphs are also of great interest for network flow problems (see, e.g.,~\cite{series-parallel-min-cost-flow,robust-flow-networks}). \Cref{thm:gwbad,thm:gwapproximation} show that greedy-weight's approximation ratio is highly linked to the width stability of the graph.
%

In \Cref{sec:mdfd} we continue with a generalised version of MFD, \emph{Minimum Cost Circulation Decomposition (MCCD)}, on directed graphs with cycles and no sinks or sources, and a cost function on the edges. Instead of decomposing a flow into weighted paths, we decompose a circulation into \emph{weighted circulations} and minimise the total cost of the circulations, and instead of the width, a natural lower bound for this problem is the minimum cost of a circulation cover (\textsf{mccc}). Decomposing into circulations rather than paths is a natural generalisation, as paths can be considered as value $1$ flows themselves.
Additionally, we also consider a relaxation in which the flow/circulation decomposition might use negative integer weights on flows/circulations, rather than strictly positive weights as has traditionally been considered~\cite{vatinlen2008simple,hartman2012split,kloster2018practical}.
An important observation that we leverage for this variant (unlike the positive-only version)
is that the \textsf{width}/\textsf{mccc} stays constant as flow is chosen and removed. Using this,
we give a $(\lceil \log \Vert X \Vert \rceil +1)$-approximation algorithm for this variant.


We denote all versions by \mccdpos~and \mccdint~as well as \mfdpos~and \mfdint~throughout the paper. 
While \mccdy~and \mfdint~are natural versions of the problem, they have not been previously considered in the MFD literature to our knowledge. However, \mfdint~can also have natural applications, since by applying \mfdint on the difference between two flows, one can minimally explain the \emph{differences} between them, e.g. to explain the differences in RNA expression between two tissue samples with the fewest number of up/down regulated transcripts, which is often the goal of RNA sequencing experiments \cite{teng2016benchmark}.
Our approximation follows a \emph{power-of-two} approach
where the weights of the flows/circulations chosen are (positive or negative) powers of two.
More specifically, observe that if all circulation values are even, then one can divide them by 2 and obtain a circulation $X$ with smaller $\Vert X \Vert$ whose decomposition can be transformed back into a decomposition of $X$. In order to obtain such an even circulation, we prove a basic property that can be of independent interest: given any integer circulation $X$, there exists a \emph{unitary} circulation (its values are $0$, $+1$, or $-1$) $Y$, such that $X+Y$ is even on every edge (\Cref{lem:parity-unitary}). In addition, given a unitary circulation $Y$, we show that $Y$ can be decomposed into circulations of total cost no more than \textsf{mccc} (\Cref{lem:unitary-to-difference}). We obtain the $(\lceil \log \Vert X \Vert \rceil +1)$-approximation ratio (\Cref{thm:mdfd-approx}) by iteratively removing the unitary circulation, dividing all circulation values by 2, and preprocessing the graph so that the \textsf{mccc} is a lower bound on the size of the \mccdint. Summarised, we show: 
\begin{restatable}{theorem}{mdfdapprox}
\mccdint~can be approximated with a factor of $\log \lceil \Vert X \Vert \rceil + 1$ in runtime $O(n\log m(m+n\log n)+m\log\Vert X\Vert)$.
\label{thm:mdfd-approx}
\end{restatable}
By~\Cref{lem:mcfdint-to-mfdint} we additionally obtain the result for \mfdint. Notably, the runtime of the algorithm does not depend on the cost function.

Finally, in \Cref{sec:k-flow-weight-assignment} we consider a closely related problem, called \emph{$k$-Flow Weight Assignment}~\cite{kloster2018practical}. In addition to the flow $X$, in this problem we are also given a set of $k$ paths, and we need to decide if there is an assignment of weights to the paths such that they form a decomposition of $X$. If the weights belong to $\mathbb{N}$, this was shown to be NP-complete in~\cite{kloster2018practical}.
In this work, we first observe that in the same way that allowing negative integer weights simplifies the approximability of MFD, allowing weights to belong to $\mathbb{Z}$ fully changes the complexity of the $k$-Flow Weight Assignment Problem, making it polynomial.
This is due to the fact that the linear system defined by the given paths loses its only inequality of restricting the weights to positive integers. It thus transforms an ILP to a system of linear diophantine equations, which can be solved in polynomial time (see e.g.~\cite{Shrij1986linearprog}). Second, we consider a conjecture from~\cite{kloster2018practical} stating that if the weights belong to $\mathbb{N}$, and $k$ is the size of a \mfdpos for $X$, then the problem admits a unique solution (i.e., a unique assignment of weights to the given paths). 
If true, this would speed up the FPT algorithm of~\cite{kloster2018practical} for \mfdpos, because
a step solving an ILP could be executed by solving a standard linear program returning a rational solution and
checking if the (supposedly unique) solution to this system is integer.
Moreover, the same conjecture (with the same implication) was also a motivation behind the greedy algorithm of~\cite{shao2017theory} for \mfdpos. 
In this paper, we disprove the conjecture of~\cite{kloster2018practical}, further corroborating the gap between \mfdpos and \mfdint.

\section{Preliminaries}

In \Cref{sec:greedy,sec:k-flow-weight-assignment} we are given a directed acyclic graph $G=(V,E)$. Without loss of generality, we assume a unique source $s$ and a
unique sink $t$ with no in-edges and no out-edges respectively; otherwise, the graph can be
converted to such a graph by adding a pseudo source and sink and connecting them to all sources and sinks
respectively. 
We denote by $\deg^+(v)$ and $\deg^-(v)$ the out- and indegree of a vertex $v$, respectively.
While Minimum Flow Decompositions are also studied for graphs with cycles (see,
e.g.,~\cite{vatinlen2008simple,hartman2012split}), the task is still to decompose into simple paths, and
so our inapproximability result on DAGs in \Cref{sec:greedy} also applies for graphs with cycles.
In \Cref{sec:mdfd} we consider directed graphs $G = (V,E,c)$ with no sources or sinks, where $c:E \to \mathbb{R}_{\geq0}$ is a cost function. Such graphs can not be acyclic.
We use $n$ and $m$ to denote the number of nodes and edges of $G$, respectively. 
For both kinds of graphs, we call functions $X\colon E\to\mathbb{Y}$ \emph{pseudo-flows},\footnote{Commonly in the literature, (pseudo-)flows are additionally required to be skew-symmetric and to be upper-bounded by some capacity function on the edges.}
where $\mathbb{Y}$ is some set of allowed flow values (numbers). 
We treat pseudo-flows as vectors over $E$ and use the notation $X + Y$ and $aX$ to denote the (element-wise) sum of pseudo-flows and multiplication by a scalar, respectively.
The numbers $0$ and $1$ also denote (depending on context) pseudo-flows that are $0$ (resp.~$1$) on every edge. 
We write $X \leq Y$ (and similarly $<$) to mean $X(u,v) \leq Y(u,v)$ for every $(u,v)\in E$.

Given a DAG $G$, a \emph{flow} is  a pseudo-flow
satisfying conservation of flow (incoming flow equal to outgoing flow) on internal nodes $V \setminus \{s,t\}$. A pseudo-flow satisfying the conservation of flow on all nodes is called a \emph{circulation}. We sometimes refer to the value $X(e)$ for a flow/circulation $X$ as the flow of the edge $e$.
It is known that the sum of two flows/circulations $X + Y$, the multiplication of a flow/circulation with by a scalar $aX$,
and the empty pseudo-flow $0$ are themselves flows/circulations. 
Let $|X|$ denote the total flow out of $s$ (or by flow conservation, equivalently into $t$) for a flow $X$. Note that $|X|$ can be negative.
Given an $s$-$t$ path $P$, denote by $P[u..v]$ the subpath of $P$ going from $u\in V$ to $v\in V$ and let $P$ also denote the flow defined by setting $1$ to every edge in $P$ and $0$ to every other edge. With these definitions, we are ready to formally define MFD.

\begin{definition}
Given a flow $X$, a \emph{flow decomposition} of $(G,X)$ of \emph{size} $k$ is a family of $s$-$t$ paths
$\mathcal{P}=(P_{1},\dots,P_{k})$ with weights $(w_{1},\dots,w_{k})\in\mathbb{Y}^k$ such that
$X=w_{1}P_{1}+\dots+w_{k}P_{k}$.
\end{definition}


\begin{definition}
Given a flow $X$,
let $\mfdsizey{G}{X}$ be the smallest size of a flow decomposition of $(G,X)$ with weights in $\mathbb{Y}$.
\end{definition}

We omit $\mathbb{Y}$ if it is clear from the context. We call the problem of producing a flow decomposition of $(G,X)$ of minimum size 
the \emph{minimum flow decomposition (MFD) problem}.

\begin{figure}
    \centering
    \includegraphics[width=0.3\textwidth]{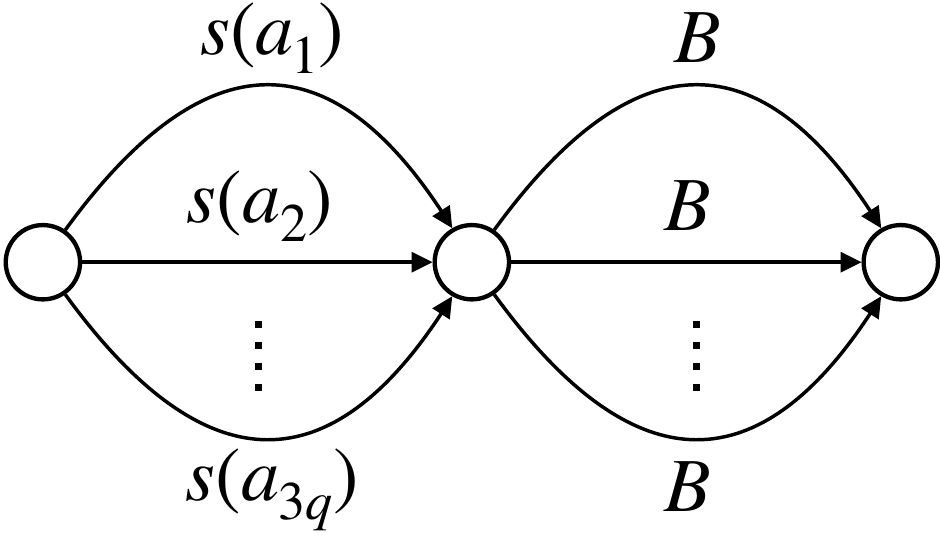}
    \caption{The reduction of $3$-Partition to MFD from~\cite{vatinlen2008simple}. The $3$-Partition instance consists of a set $A = \{a_1,\dots,a_{3q}\}$, where every $a_i$ has a positive integer \emph{size} $s(a_i)$, and a positive integer $B$, such that $B/4 < s(a_i) < B/2$ holds for every $a_i \in A$. The question is whether $A$ can be partitioned into $q$ disjoint subsets, each of $3$ elements and of size $B$. The MFD series-parallel (see \Cref{def:ser-par}) reduction consists of a subgraph obtained by the \emph{parallel composition} of $3q$ edges with flow values $s(a_1),\dots,s(a_{3q})$, and a subgraph obtained by the parallel composition of $q$ edges, each with flow value $B$. These two graphs are composed with the \emph{series composition}. Intuitively, because $B/4 < s(a_i) < B/2$ holds for every $a_i$, the MFD consists of exactly $3q$ paths of weights $s(a_1),\dots,s(a_{3q})$, and each edge on the right-hand subgraph is traversed by exactly three paths whose weights sum to $B$, giving thus the partition of $A$. Moreover, since the first $3q$ edges need to be decomposed, the previous decomposition is minimum even if negative weights are allowed, making \mfdint NP-hard.
    }
    \label{figure:np-hardness-reduction}
\end{figure}

In \Cref{sec:greedy} we study \mfdpos ($0 \in \mathbb{N}$), and in \Cref{sec:mdfd} we study \mfdint and its generalisation \mccdint.
Note that the reduction showing \mfdpos to be strongly NP-hard from \cite{vatinlen2008simple}
also holds for \mfdint (see~\Cref{figure:np-hardness-reduction}). However, a flow with non-negative values may admit a decomposition using
fewer paths if negative weights are allowed, as shown in~\Cref{fig:more-pos-paths}. We explore 
further differences between \mfdpos and \mfdint in \Cref{sec:mdfd,sec:k-flow-weight-assignment}.


\begin{figure}
  \begin{subfigure}[t]{.45\linewidth}
    \centering\includegraphics[width=.5\linewidth]{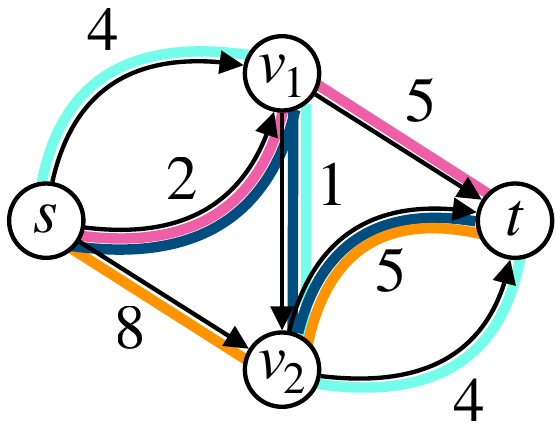}
    \caption{If negative weights are allowed, the four paths decompose the flow with weights $4,5,8,$ and $-3$ (dark blue).}
    \label{fig:more_pos_paths:neg}
  \end{subfigure}
  \hfill
 \begin{subfigure}[t]{.45\linewidth}
    \centering\includegraphics[width=.5\linewidth]{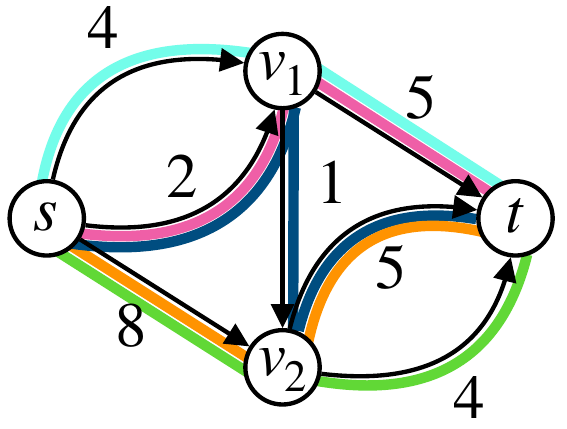}
    \caption{With positive weights only, five paths are needed, since the edge $(v_1, v_2)$ must
    be decomposed by a weight $1$ path, leaving 4 edges that must be covered separately. The paths shown are one such decomposition.}
    \label{fig:more_pos_paths:pos}
  \end{subfigure}
\caption{A positive flow admitting a decomposition into four paths only if negative weights are allowed.}
\label{fig:more-pos-paths}
\end{figure}



Let $\Vert X\Vert=\max_{(u,v)\in E}|X(u,v)|$ denote the infinity norm on flows or circulations.
In particular, notice that if $\mathbb{Y} \subseteq \mathbb{Z}$,
then $\Vert X \Vert \leq 1$ means that $X(u,v)\in \{0,\pm 1\}$ for every $(u,v)\in E$.
Let $X\equiv_{2}Y$ iff $X$ and $Y$ have the same parity everywhere, i.e., for every $(u,v)\in E$, we have that $X(u,v)$ is odd iff $Y(u,v)$ is odd.

\begin{definition}
Given $S \subseteq E$, we define $\textsf{width}_S(G)$ as the minimum number of $s$-$t$ paths in a DAG $G$ needed to cover all edges of $S$. If $S = E$ we just write $\width{G}$.
\end{definition}

The width is the main combinatorial tool that we use for our approximation results, and we will show in \Cref{sec:greedy} that it is highly linked to the approximation performance of greedy-weight. 
Just like its more common node variant, $\width{G}$ can be computed
in $O(mn)$ time. As described by, e.g., \cite{ahujia1993network,Caceres:2021wi}, 
this is done by reduction to a min-flow instance with demand one on every edge; 
the minimum flow of this instance is $\width{G}$, and the flow can be found by reduction
to a max-flow instance.
Moreover, the problem can be relaxed to only require the coverage of $S\subseteq E$ and solved in the same running time by setting the demands only on the edges of $S$.

\begin{lemma}[\cite{ahujia1993network,max-flow-orlin}]
\label{lem:min-flow-cover}
Let $G = (V, E)$ be a DAG, and $S\subseteq E$. A flow $C: E \to \mathbb{N}$ can be computed in $O(mn)$ time, such that $C(e)\geq1$ for all $e\in S$ and $|C|=\textsf{width}_S(G)$.
\end{lemma}

The flow $C$ with total flow $\textsf{width}_S(G)$ suffices and we do not need to calculate a path cover achieving that minimum. However, we note that it can be directly computed given the flow $C$.
We can think of this path cover as a
flow decomposition of $C$ into $\textsf{width}_S(G)$ weight-one paths, which can be
found by greedily removing such paths from $C$ until it is completely decomposed. Since
each path has no more than $n-1$ edges and since $\textsf{width}_S(G) \leq m$, the overall
runtime of finding the path cover is $O(mn)$. Similarly, every path cover $P_1,\dots,P_\ell$ of $G$ defines a flow $C$ on $G$: $C = P_1 + \dots + P_\ell$, and we say that $C$ is the induced flow of the path cover $(P_i)_i^\ell$.

\begin{definition}
    In a directed graph $G = (V,E)$ we call a subset $C \subseteq E$ an \emph{antichain} of $G$ if all edges in $C$ are pairwise parallel, that is there exists for no pair of edges in $C$ a path in $G$ leading from one edge to the other.
\end{definition}
It can be shown with straight forward arguments that $\width{G} = |C|$ for a maximum (sized) antichain $C$ of $G$.

\section{Width matters for greedy approaches} \label{sec:greedy}

Since the difference of two flows is still a flow, it is very natural to consider
successively removing
the simplest type of flow --- that is to say, paths --- as an algorithmic strategy for \mfdpos.
Indeed, the particular greedy path removal strategy of finding a \emph{heaviest path} (\emph{greedy-weight})
is commonly used as a heuristic in applications (e.g.,~\cite{pertea2015stringtie,baaijens2020strain,tomescu2013novel,hartman2012split}) and it
seems to be mentioned in nearly every paper addressing flow decomposition. 
More formally, a path $P$ is said to \emph{carry} flow $p$ if $X(u,v) \geq p$ for all edges $(u,v)$ of $P$ (in particular, a $v$-$v$ path carries infinite flow). A \emph{heaviest path} is an $s$-$t$ path carrying the largest flow. Such a path can be easily found in linear time in the size of the DAG by dynamic programming (see, e.g.,~\cite{vatinlen2008simple}).


\subsection{Width hinders greedy on \mfdpos}\label{sec:greedy-lb}

We define a family of \mfdpos~instances $(G_{\ell},X_{\ell,B})$, depending on two parameters $\ell \in \mathbb{N}\setminus \{0\}$ and $B\in \mathbb{N}$. The family is defined recursively on $\ell$. The base case $(G_1,X_{1,B})$ for $\ell = 1$ is shown in \Cref{fig:greedy:base_case}. For $\ell > 1$, we build $(G_\ell,X_{\ell,B})$ from two disjoint copies of $(G_{\ell-1},X_{\ell-1,B})$, by adding 5 extra edges and flow values as shown in \Cref{fig:greedy:recursion}.  We call the edge connecting the two copies of $G_{\ell-1}$ a \emph{central edge}. Edges whose flow value depends on $B$ are called \emph{backbone} edges, and they form a $s$-$t$ path.
By choosing $B=2^{\ell+1}$, we show that the flow $X_{\ell, 2^{\ell+1}}$ can be decomposed using a number of paths linear in $\ell$, thanks to the heavy backbone edges, whereas the greedy-weight algorithm fully saturates the central edges with its first path and is left with a remaining flow requiring $2^{\ell+1}$ paths to be decomposed.

\begin{figure}[t]
  \begin{subfigure}[t]{.3\linewidth}
    \centering\includegraphics[width=.9\linewidth]{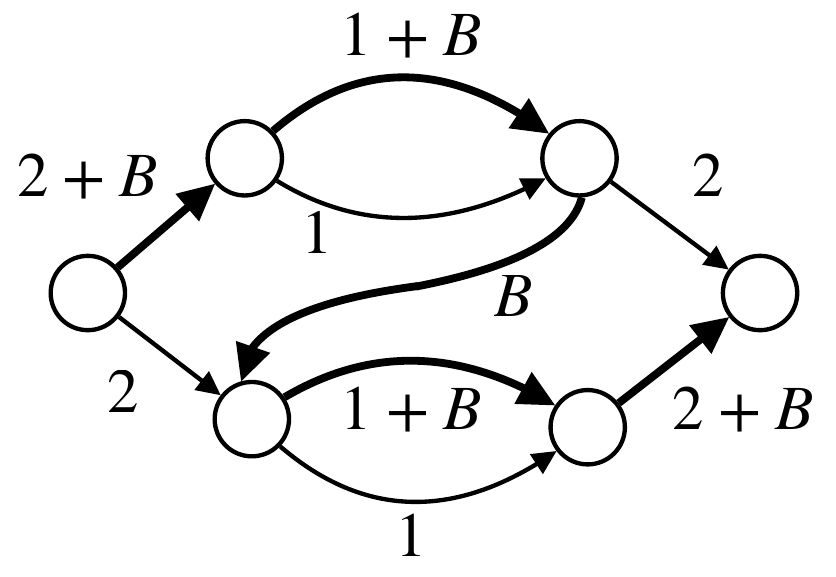}
    \caption{The base case $(G_1,X_{1, B})$ ($\ell = 1$). Backbone edges (bold) carry $B$ flow.}
    \label{fig:greedy:base_case}
  \end{subfigure}
  \hfill
 \begin{subfigure}[t]{.3\linewidth}
    \centering\includegraphics[width=.9\linewidth]{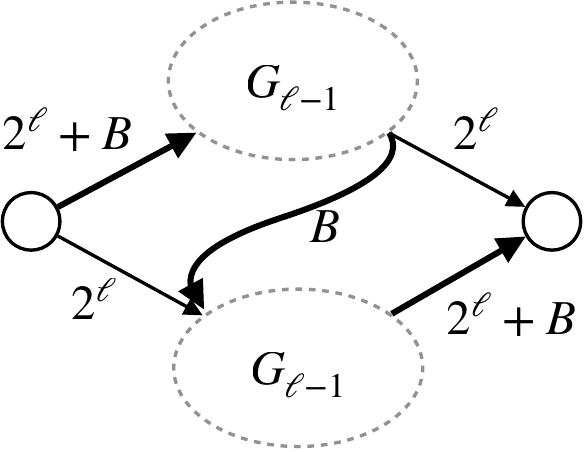}
    \caption{Building $(G_\ell, X_{\ell, B})$ from two copies of $(G_{\ell-1}, X_{\ell-1,B})$ ($\ell > 1$). 
    The $5$ edges entering (resp. leaving) $(G_{\ell-1}, X_{\ell-1,B})$ are defined to enter (resp. leave) the source (resp. sink) node of $(G_{\ell-1}, X_{\ell-1,B})$. The central edge has flow $B$ and is part of the backbone (bold edges).}
    \label{fig:greedy:recursion}
  \end{subfigure}
  \hfill
  \begin{subfigure}[t]{.3\linewidth}
    \centering\includegraphics[width=.9\linewidth]{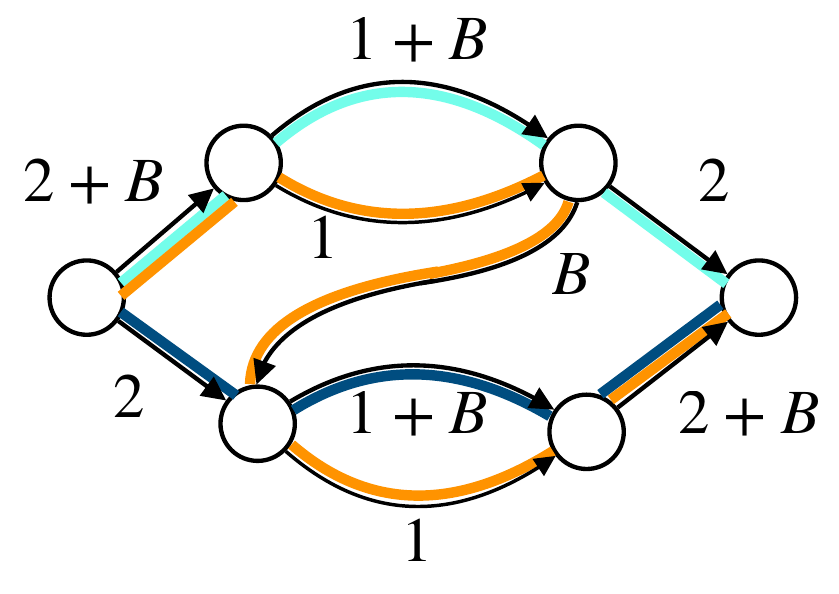}
     \caption{
     Decomposing the base case $(G_1, X_{1, B})$ ($\ell = 1$), for $B = 2^{\ell+1}$. All non-backbone edges can be decomposed with $2\ell + 1=3$ paths. The orange 
     path has weight $1$ and dark and light blue paths have weight $2$. A fourth path (of weight $3$) along the backbone is required
     to fully decompose the flow.
     }
     \label{fig:greedy:base_case_paths}
  \end{subfigure}
\caption{Construction for $(G_\ell, X_{\ell,B})$. Setting $B=2^{\ell+1}$ yields \mfdpos instances where greedy-weight uses
$\Theta(\frac{m}{\log m})$ times more paths than optimal to decompose the flow.} \label{fig:greedy}
\end{figure}

\begin{lemma}
\label{lem:gw-bad-greedy}
Let $G_{\ell}$ with flow $X_{\ell,2^{\ell+1}}$ be constructed as described before. Greedy-weight uses $1+2^{\ell+1}$ paths to decompose
$X_{\ell,2^{\ell+1}}$.
\end{lemma}
\begin{proof}
We first show that the heaviest path in $X_{\ell,2^{\ell+1}}$ follows every backbone edge from $s$
to $t$. Certainly this path carries $2^{\ell+1}$ flow, and no more, since every backbone edge has flow at least $2^{\ell+1}$ and every central edge has flow exactly $2^{\ell+1}$. To see that there is no heavier path, observe that all the non-backbone edges have flow value strictly less than $2^{\ell+1}$ by construction.
After removing that path
with weight~$2^{\ell+1}$, all the central edges are completely decomposed.
The remaining graph and flow, without central edges, has $2^{\ell+1}$ weight-$1$ edges  that are pairwise non-reachable using only non-zero flow edges (4 edges for each of the $2^{\ell-1}$ copies of $G_1$), each of which must be covered by a different path of weight $1$ (and these paths fully decompose the flow).
\end{proof}


\begin{figure}
     \centering\includegraphics[width=.45\linewidth]{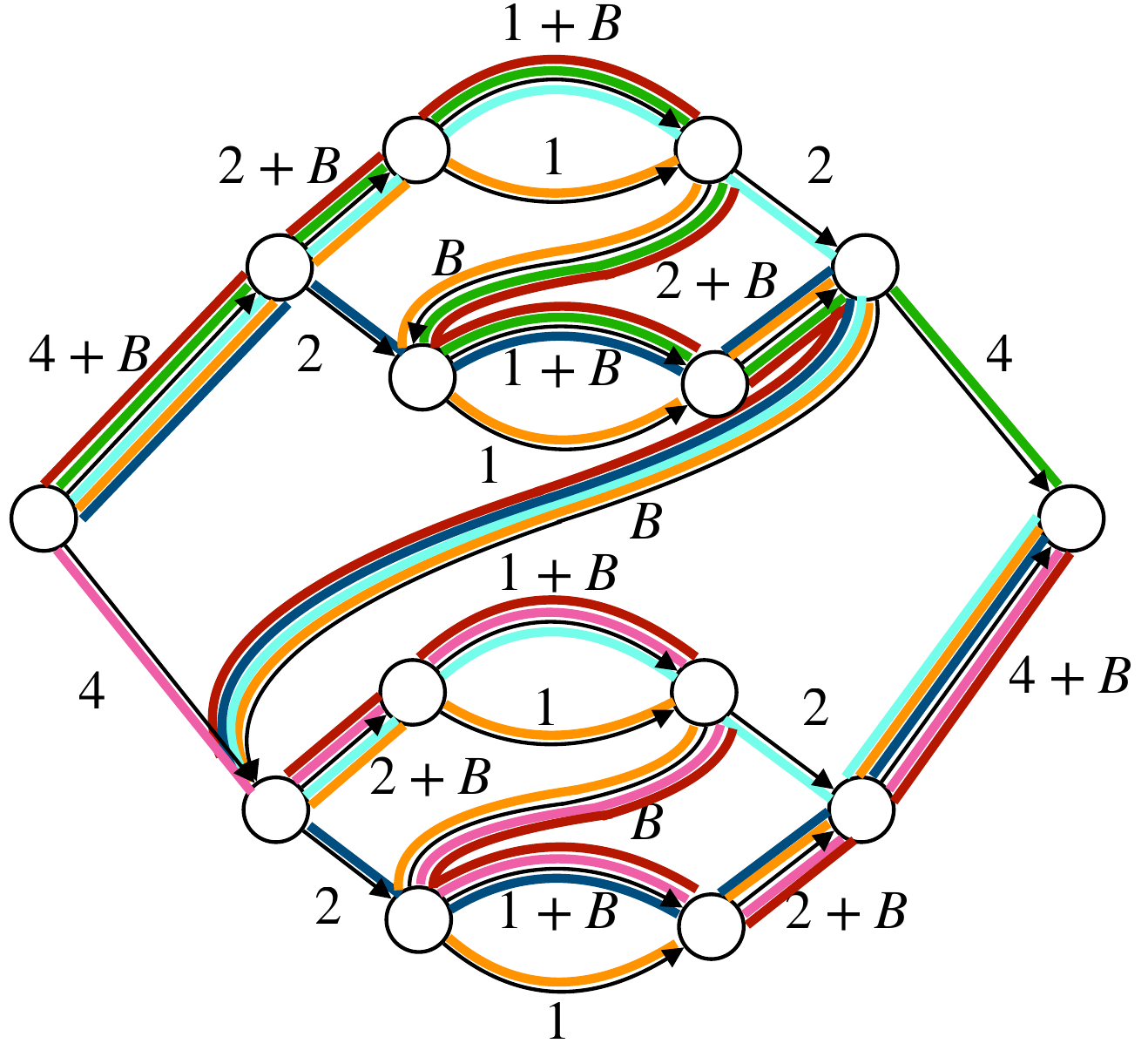}
    \caption{$(G_2,X_{2, B = 2^{\ell+1}})$ with a $2\ell+2=6$ flow decomposition. Orange (weight 1) and dark and light blue (weight 2)
    decompose the non-backbone edges in the two copies of $G_1$. Green and pink (weight 4) decompose the
    non-backbone edges added in $G_2$. Dark red (weight $3$) is the additional path needed to fully decompose the flow.
    \label{fig:greedy:opt}}
\end{figure}

\begin{lemma}
\label{lem:gw-bad-lb}
Let $G_{\ell}$ with flow $X_{\ell,2^{\ell+1}}$ be constructed as described before. It is possible to decompose $X_{\ell,2^{\ell+1}}$ using $2\ell + 2$ paths.
\end{lemma}
\begin{proof}
Using induction, we first prove that we can use $2\ell+1$ paths to decompose all of the flow of $X_{\ell,2^{\ell+1}}$ on the
non-backbone edges of $G_\ell$, and that these paths have total weight $2^{\ell+2} - 3$.
When $\ell=1$ (see \Cref{fig:greedy:base_case_paths}), we can decompose both the flow-$1$ edges with a single path of weight $1$ which goes through the central edge. Moreover, we can decompose the flow-$2$ edges with two paths of weight $2$, without using the central edge.
These paths have $5=2^{1+2}-3$ total flow.  We now assume that the claims hold for $\ell=k$ and prove it for $\ell=k+1$.
Consider the graph $G_{k+1}$. By assumption, the non-backbone edges in every copy of $G_k$ are fully decomposed by $2k+1$
paths, and the total flow of those paths is $2^{k+2}-3$. Note that these paths are $s$-$t$ paths in $G_k$, but they must be extended to be $s$-$t$ paths in $G_{k+1}$. Because there is a central edge with weight $2^{(k+1)+1}=2^{k+2} > 2^{k+2} -3$ from one copy of $G_k$ to the other,
it is possible to route the $2k+1$ paths from
the first $G_k$ to the second using the central edge. Additionally, the backbone edges from $s$ to the first copy of $G_k$
and from the second copy of $G_k$ to $t$ have flow $2^{k+1} + 2^{k+2}$, so they can be used to complete the routing of the paths
from $s$ to $t$.
Then we can use two additional paths of weight $2^{k+1}$  each to decompose the flow of the two new non-backbone edges; one using the backbone path of the upper $G_k$ (extended by the upper edges of $G_{k+1}$), and analogously, the other through the backbone path of the lower $G_{k}$, which is possible since the paths obtained inductively only decompose $2^{k+1}$ of the backbone flow of $G_{k}$ whereas now $2^{k+2}$ backbone flow must be decomposed. As such, we use $2k+1+2=2(k+1) +1$ paths with total flow 
$2^{k+2}-3 + 2(2^{k+1})= 2^{(k+1)+2} - 3$,
as required.
By the previous, given any graph $G_\ell$, we can decompose all non-backbone edges using $2\ell+1$ paths.
Note that removing these weighted paths from $X_{\ell,2^{\ell+1}}$ yields a flow on $G_\ell$ (of value $3$).
Because the remaining edges (all backbone) form a path from $s$ to $t$, and the remaining
edge values form a flow on $G_\ell$, all remaining edges must have the same flow value, and can be covered by
one path. Thus, $2\ell+2$ paths are sufficient to decompose $X_{\ell,2^{\ell+1}}$. See \Cref{fig:greedy:opt}
for an example when $\ell=2$.
\end{proof}




\ratiomfdpos*


\begin{proof}
By \Cref{lem:gw-bad-greedy,lem:gw-bad-lb}, greedy-weight uses $\Theta(2^{\ell})$ paths to decompose the flow $X_{\ell,2^{\ell+1}}$ described
above, whereas it is possible to decompose the flow with only $\Theta(\ell)$ paths. It can be easily verified by induction that the number of edges of $G_\ell$ is $7 \cdot 2^{\ell} -5$. 
So the ratio between greedy-weight and the optimal for this instance is $\Omega(\frac{m}{ \log m})$.
\end{proof}


\begin{figure}
    \centering\includegraphics[width=.6\linewidth]{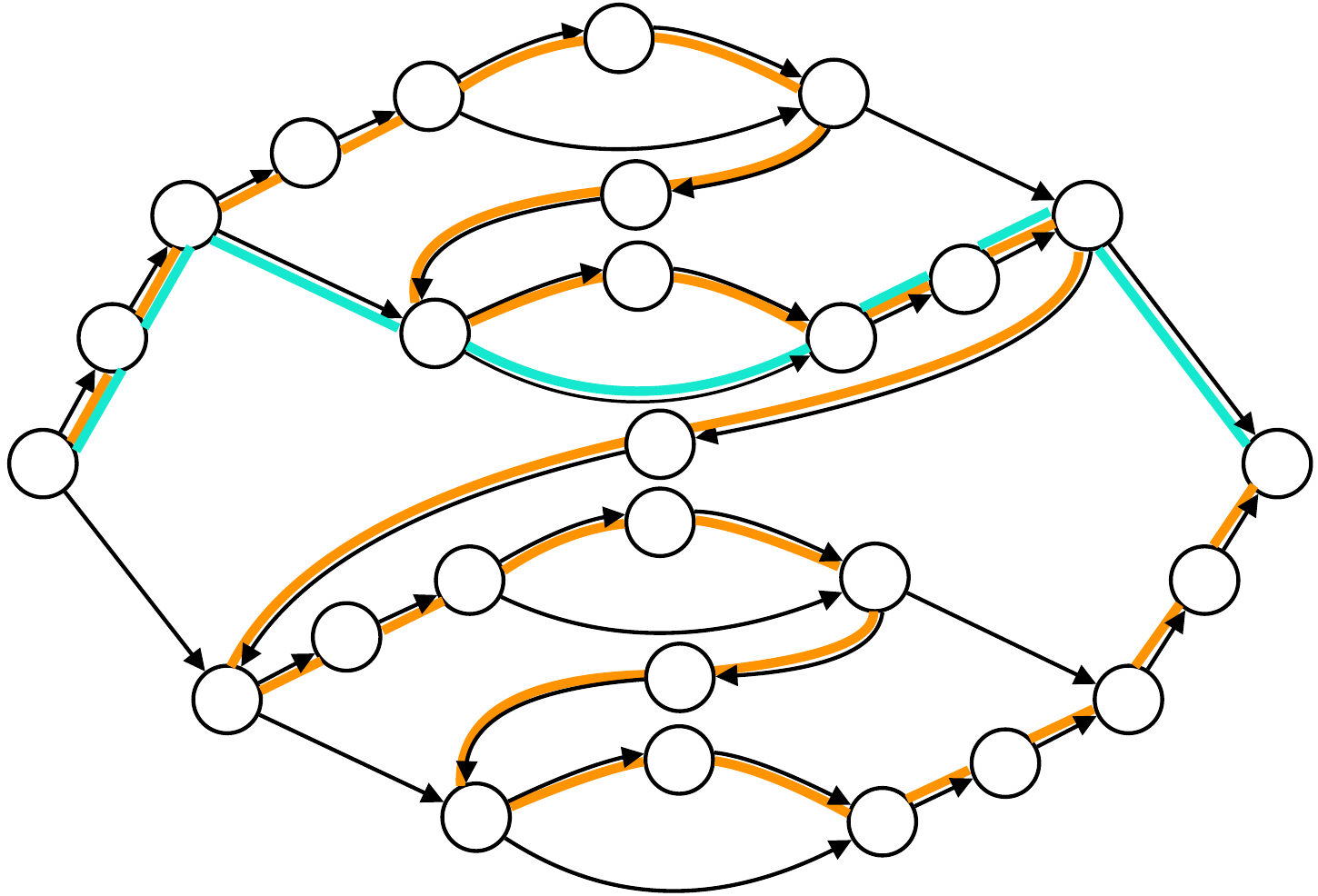}
    \caption{By subdividing the backbone edges, choosing longest paths (orange) and shortest paths (blue) both give the same approximation ratio as choosing heaviest paths in the original flow $X_{\ell,2^{\ell+1}}$.}
    \label{fig:greedy:g_star}
\end{figure}

While greedy-weight is most commonly used in applications,
the approach was first presented as part of a general framework~\cite{vatinlen2008simple}: pick any optimality criteria
for $s$-$t$ paths that is \emph{saturating} (i.e., fully decomposes at least one edge), and successively remove
optimal paths. Since each path is saturating, the algorithm must decompose the flow in $m$ or fewer paths.
Another optimality criterion sometimes used in DNA assembly (e.g., in vg-flow~\cite{baaijens2020strain})
is the \emph{longest path} (with its maximum possible flow so that it is saturating). To adapt our construction of $(G_\ell, X_{\ell,2^{\ell+1}})$ so that this approach yields the same approximation ratio, consider $(G_\ell^*, X_{\ell, 2^{\ell+1}}^*)$,
constructed as in $(G_\ell, X_{\ell,2^{\ell+1}})$ except that we split every backbone edge $(u,v)$ into two edges,
$(u,w)$ and $(w,v)$. See \Cref{fig:greedy:g_star} for an example.
Then the path along the backbone edges will be the longest from $s$ to $t$ and the previous asymptotic analysis still holds,
since we no more than doubled the number of edges (and the number of edges of the new construction is still $\Theta(2^\ell)$).
Yet another optimality criterion, studied in~\cite{hartman2012split} for its application to network routing,
is the shortest path (again with its maximum possible
flow). $(G_\ell^*$, $X_{\ell, 2^{\ell+1}}^*)$ will also force this approach to take an exponential number of paths, since first the algorithm will
decompose all $2^{\ell+1}$ weight-$1$ edges with $2^{\ell+1}$ different paths.

\subsection{Greedy approximation for width-stable graphs}

As exploited in Section~\ref{sec:greedy-lb}, one sticking point for greedy path removal algorithms is the fact that the width of a graph can increase after an edge is fully decomposed.
We now identify a new class of graphs, in which the graph does not increase its width during the execution of the algorithm. We show that greedy-weight decomposes ``enough'' flow at each step in these graphs, giving a \gwsimple-approximation for \mfdpos.

If $X \geq 0$ is a flow on a DAG $G$,
we write $G|_{X}$ ($G$ \emph{restricted to} $X$) to mean the spanning subgraph of $G$ made up of the edges $e \in E$
such that $X(e) \neq 0$. Conversely, if $S$ is a subgraph of $G$, we
write $X|_{S}$ ($X$ \emph{restricted to} $S$) to mean the pseudo-flow $X$ only on the edges of $S$.
In the case of \mfdpos, once an edge is fully decomposed, it cannot be used in future paths,
possibly increasing the width of the graph that can be used to decompose the rest of the flow and
sometimes triggering an increase of the size of a minimum flow decomposition as well. We call a graph \emph{width-stable} 
if it does not have this issue.


\begin{definition}[Width-stable graph] \label{prop:width}
We say that a graph $G$ is \emph{width-stable} if, for any non-negative flows $X \leq Y$ on $G$, it holds that $\width{G|_X} \leq \width{G|_Y}$.
\end{definition}



Many useful \mfdpos instances satisfy \Cref{prop:width}. For example, 
the first proof of MFD's NP-hardness~\cite{vatinlen2008simple} was a reduction to a very simple graph of this form,
as shown in \Cref{figure:np-hardness-reduction}; this means that \mfdpos restricted
to width-stable graphs is also NP-hard. 

\begin{definition}[\cite{garlet2020efficient}]
    We call an $s$-$t$ DAG $G$ \textit{funnel} if every $s$-$t$ path has a private edge that is not contained in any other $s$-$t$ path. 
\end{definition}

Funnels are simple graphs in the sense that they admit a unique flow decomposition~\cite{khan2022improving}. We use funnels in \Cref{lem:enough-flow} to characterize graphs that are not width-stable.
Funnels generalise in/out-forests: along any $s$-$t$ path nodes $v$ first satisfy $\deg^-(v) \leq 1 \leq \deg^+(v)$ and then $\deg^-(v) \geq 1 \geq \deg^+(v)$. We call a node $v$ \textit{forking} (resp. \emph{merging}) if $\deg^+(v) > 1$ (resp. if it is $\deg^-(v) > 1$). 
For a funnel subgraph $F$ of $G$ we call a path in $G$ from a merging node in $F$ to a forking node in $F$ a $\textit{central path}$ of the funnel. The graphs $(G_\ell, X_{\ell, B})$ in \Cref{sec:greedy-lb} are precisely funnels with central paths.

The next property that we need is that there is always, during the execution of the greedy-weight algorithm, a path carrying ``enough'' flow
from $s$ to $t$. 


\begin{figure}
\centering
    \includegraphics[width=.3\textwidth]{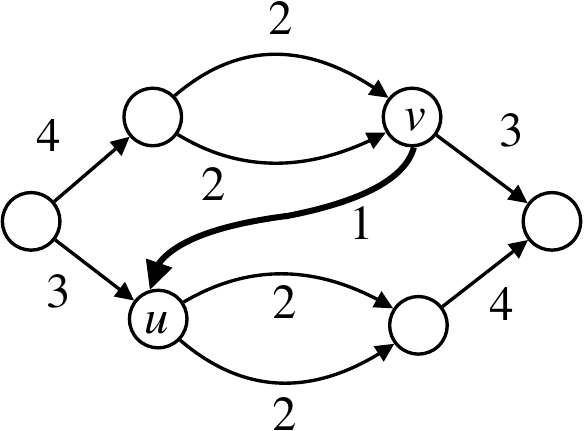}
    \caption{\mfdpos instance with $\width{G}=3$, $|X|=7$, $|X|/\width{G}=7/3>2$, but no path can carry more flow than $2$. By \Cref{lem:enough-flow}, this is equivalent to $G$ not being width-stable, there exists an $s$-$t$ path saturating the central edge from $v$ to $u$ of weight $1$, making the graph a funnel and increasing $\width{G}$ to $4$.}
    \label{fig:widest_path}
\end{figure}

\begin{lemma}
\label{lem:enough-flow}
Let $G$ be an $s$-$t$ DAG. The following statements are equivalent:

\begin{enumerate}
    \item $G$ is width-stable, \label{lem:enough-flow-stable}
    \item $G$ has \textit{paths of large weight}: for any flow $X\ge0$ on $G$, there exists an $s$-$t$ path in $G|_X$ carrying $|X|/\width{G|_X}$ flow, \label{prop:enough-flow}
    \item $G$ has no funnel subgraph with a central path. \label{lem:enough-flow-no-backbone}
\end{enumerate}
\end{lemma}

See~\Cref{fig:widest_path} for an example.

\begin{proof}
$(\ref{lem:enough-flow-stable}) \implies (\ref{prop:enough-flow})$: Let $G=(V,E)$ be an $s$-$t$ DAG with flow $X \geq 0$ for which there is no $s$-$t$ path in $G|_X$ carrying $b = |X|/\width{G|_X}$ flow. We will show that $G$ is not width-stable. For simplicity, we assume $G = G|_X$, the result follows immediately for any supergraph of $G$.
Let $S$ be the set of vertices reachable from $s$ by a path carrying $b$ flow. By assumption, $t \notin S$, so $(S, V\setminus S)$ is an $s,t$-cut and $X(u,v) < b$ for any edge $(u,v)$ with $u \in S, v \in V \setminus S$. Define $T$ to be the set of vertices that can reach $t$ via a path involving vertices only from $V \setminus S$:
\[ T \coloneqq V \setminus \{ v \in V \mid \text{all $v-t$ paths cross $S$} \}. \]
Since $S \subseteq V \setminus T$, $(V \setminus T, T)$ is an $s,t$-cut. Note that if $(u,v)$ is an edge with $u \in V \setminus T$ and $v \in T$, then also $u \in S$ and $v \in V \setminus S$ and thus $X(u,v) < b$.
The set $C$ defined by
\[ C \coloneqq \{ (u,v) \in E \mid u \in V \setminus T, v \in T \} = E \cap (S \times T) \]
is an $s,t$-cut set (a set of edges that every $s$-$t$ path has to cross) and we have $X(e) < b$ for every $e \in C$. This implies $|C| > \width{G}$.

We construct a new flow $Y \geq 0$ on $G$ for which $C$ is an antichain in $G|_Y$. Initially, $Y \coloneqq X$. Define the following paths:
\begin{itemize}
    \item $p_s(u)$ for all $u \in S$: an $s$-$u$ path with all nodes in $S$,
    \item $p_t(v)$ for all $v \in T$: a $v$-$t$ path with all nodes in $T$.
\end{itemize}
We keep the \textbf{invariant} that the paths exist and do not change in $G|_Y$ throughout the construction of $Y$. Let denote a path from $v \in T$ to $u \in S$, all of whose internal nodes are from $V \setminus (S \cup T)$, as $b(v,u)$, and note that such a path exists iff $C$ is not an antichain. 
Assume it carries flow $\mu > 0$ (i.e., $\mu$ is the minimum flow along $b(v,u)$), 
we then perform the following operation on $Y$:
\[
Y(e) = 
\begin{cases}
    Y(e) - \mu & \text{if } e \in b(v,u)\\
    Y(e) + \mu & \text{if } e \in p_s(u) \cup p_t(v).
\end{cases}
\]
Note that after the operation $Y$ remains a flow and that none of the three paths have pairwise intersecting edges. The process does not violate the \textbf{invariant} and repeating eventually it destroys all paths from $T$ to $S$, making $C$ an antichain in $G|_Y$. This shows that $G$ is not width-stable: $Y \leq X+Y$ and $\width{G|_Y} > \width{G|_{X+Y}} = \width{G}$.
$(\ref{prop:enough-flow}) \implies (\ref{lem:enough-flow-no-backbone})$: Assume that $G$ has a funnel subgraph $F$ with a central path $P$. Let $C$ be a maximum antichain of $F$, and note that every maximum antichain of a funnel consists of private edges only. We define flows $Z_\text{stable}$, $Z_{unstable}$ and $Z$:
\begin{itemize}
    \item $Z_\text{stable}$ is defined to be the flow induced\footnote{Recall that the induced flow $X$ of a path cover $(P)_1^\ell$ is defined as $X = P_1 + \dots + P_\ell$, where we identify each path $P_i$ with a $0/1$-flow with value $1$ on the path edges (i.e., precisely the characteristic function of $P_i$).} by the minimum path cover of $F$\footnote{Since funnels admit a unqiue flow decomposition, they also admit a unique minimum $s$-$t$ path cover.} (i.e., $G|_{Z_\text{stable}} = F$).
    \item $Z_\text{unstable}$ is defined to be the flow induced by the following path cover of the graph consisting of $F$ and $P$: One $s$-$t$ path goes along $P$, covering two edges in $C$\footnote{Indeed, this path covers \emph{exactly} two edges in $C$: one edge needs to be covered to reach $P$ and another edge must be covered to reach $t$, and since $G$ is a DAG, we can not use the same edge twice.}, and the other paths cover $F$ avoiding $P$, this is possible with additional $|C| - 2$ paths.
    \item Finally, $Z \coloneqq Z_\text{stable} + Z_\text{unstable}$.
\end{itemize}
We have $|Z| = |Z_\text{stable}| + |Z_\text{unstable}| = |C| + (|C| - 1) = 2|C| - 1$ and $\width{G|_Z} \leq |C|-1$, and thus  
\[ \frac{|Z|}{\width{G|_Z}} \geq \frac{2|C|-1}{|C|-1} > 2, \]
but all $s$-$t$ paths in $G$ carry no more flow than $2$, because $C$ is an $s,t$-cut set of $G$ and all $Z$ flow values on $C$ are $2$.

$(\ref{lem:enough-flow-no-backbone}) \implies (\ref{lem:enough-flow-stable})$: Assume that $G=(V,E)$ is not width-stable, let $Y \geq X \geq 0$ be flows on $G$ with $\width{G|_X} > \width{G|_Y}$ and let $C'$ be a maximum antichain of $G|_X$. Let $F$ be a funnel subgraph of $G|_X$ containing $C'$ as maximum antichain, and let $C$ be the rightmost maximum antichain of $F$ (that is, the head of every edge in $C$ is $t$ or is merging). Since $C$ is not an antichain of $G|_Y$, there must be a path $P$ in $G|_Y$ connecting two edges in $C$, and it starts at a merging node. It must also enter a forking node, because otherwise adding the path would not decrease the width. This shows that a prefix path of $P$ is a central path of $F$. 

\end{proof}




\begin{lemma}\label{lem:gw-paths}
Let $G = (V, E)$ be a width-stable graph, $\width{G} \ge 2$. Greedy-weight uses at most $\lfloor \log |X| / \log \frac{\width{G}}{\width{G}-1} \rfloor + 1$ paths to decompose any flow $X: E \to \mathbb{N}$.
\end{lemma}

\begin{proof}
Let $b = \width{G}$. Since $G$ is width-stable, greedy-weight removes a path of weight at least $|X'|/b$ at every step by \Cref{lem:enough-flow}, where $X'$ is the remaining flow of the corresponding step. As such, after $c$ steps $|X'| \le |X|\left(\frac{b-1}{b}\right)^c$. If $|X|\left(\frac{b-1}{b}\right)^c < 1$, then $|X'| = 0$, since $|X|$ and the weights of the removed paths belong to $\mathbb{N}$. Solving for $c$ we obtain $c > \log |X| / \log \frac{b}{b-1}$. Therefore, greedy-weight takes (uses) at most $c = \lfloor \log |X| / \log \frac{b}{b-1} \rfloor + 1$ steps (paths).
\end{proof}

\gwapprox*

\begin{proof}
Assume $X>0$ (otherwise, replace $G$ by $G|_X$). Thus, $b=\width{G} \leq \mfdsizepos{G}{X}$, since any flow-decomposition of $X$ induces a path cover of $E$.
If $b\le 1$ greedy-weight finds an optimal solution. Otherwise $b\ge 2$, and \Cref{lem:gw-paths} implies that greedy-weight is a $O(\frac{\log |X| }{b \log \frac{b}{b -1}})= $\gwsimple-approximation for \mfdpos ($b \log \frac{b}{b - 1} = O(1)$ for $b \geq 2$).
\end{proof}

Finally, we show that series-parallel graphs are width-stable, and thus greedy-weight is a $O(\log{|X|})$-approximation on them.

\begin{definition}[Series-parallel graph~\cite{EPPSTEIN199241}]\label{def:ser-par}
A graph is a \emph{two-terminal series-parallel} (\emph{series-parallel} for short) graph with terminal nodes $s$ and $t$ if:
\begin{itemize}
    \item it consists of a single edge directed from $s$ to $t$, and no other nodes, \emph{or}
    \item it can be obtained from two (smaller) two-terminal series-parallel graphs $G_1$ and $G_2$, with terminal nodes $s_1,t_1$, and $s_2,t_2$, respectively, by either
    \begin{itemize}
        \item identifying $s = s_1 = s_2$ and $t = t_1 = t_2$ (\emph{parallel composition} of $G_1$ and $G_2$), or
        \item identifying $s = s_1$, $t_1 = s_2$, and $t = t_2$ (\emph{series composition} of $G_1$ and $G_2$).
    \end{itemize}
\end{itemize}
\end{definition}


\begin{corollary}
\label{cor:greedyseriesparallel}
Greedy-weight is a \gwsimple -approximation for \mfdpos~on series-parallel graphs.
\end{corollary}
\begin{proof}
%
%
%
Using~\Cref{thm:gwapproximation}, it remains to prove that any series-parallel graph $G = (V, E)$ with any flow $X: E \to \mathbb{N}$ is width-stable. We prove it using structural induction. The base case is when $G$ is single edge from $s$ to $t$, and they are trivially width-stable.

Suppose now that $G$ is obtained by the composition of series-parallel graphs $G_1, G_2$, and let $X \leq Y$ by any non-negative flows on $G$. Let $X_i = X|_{G_i}, Y_i = Y|_{G_i}$ and let $x_i$ denote $\width{G_i|_{X_i}}$, $y_i$ denote $\width{G_i|_{Y_i}}$, for $i=1,2$. 
$x_i \leq y_i$ for $i=1,2$.

If the composition operation is parallel composition, then $\width{G}=\width{G_1}+\width{G_2}$ (since edges of $G_1$ cannot reach edges of $G_2$, and vice versa) and $\width{G|_X}=x_1 + x_2$, $\width{G|_Y}=y_1 + y_2$ (since $G|_X$ and $G|_Y$ are also series-parallel), and $G$ is width-stable by the inductive hypothesis that $x_i \leq y_i$, for $i=1,2$. 

If the composition operation is series composition, width-stability follows analogously to the parallel composition by replacing sum with maximum. 
\end{proof}

However, note that there are width-stable graphs that are not series-parallel.


\section{Width helps solve \mccdint}
\label{sec:mdfd}
In this section we give an approximation algorithm for $\mfdsizeint{G}{X}$. We will obtain this for a more general problem variant, which can be defined as follows. We are given directed a graph $G = (V,E,c)$ with no sink or source nodes, with (cost) a function $c: E \to \mathbb{R}_{\geq 0}$. The \emph{cost} of a circulation $X$ is defined as $\flowcost{X} = \sum_{e \in E} c(e)X(e)$. Note that $\flowcost{\cdot}$ is a linear function: $\flowcost{X+Y} = \flowcost{X} + \flowcost{Y}$ for any two circulations $X,Y$ on $G$.

\begin{definition}
    Given $(G,X)$ of a graph $G = (V,E,c)$ and a circulation $X:E \to \mathbb{Y}$, a \textit{circulation decomposition} of size $k$ of $(G,X)$ is a family of circulations $Y_i: E \to \mathbb{N}$ with weights $(w_1,\dots,w_k) \in \mathbb{Y}^k$ such that $X = w_1 Y_1 + \dots + w_k Y_k$. We call the problem of finding a circulation decomposition of minimum cost $\flowcost{Y_1 + \dots + Y_k}$ the \textit{Minimum Cost Circulation Decomposition} or \mccdy and we write $\mccdsizey{G}{X}$ for the minimum cost.
\end{definition}

Decomposing into \emph{non-negative} weighted circulations rather than paths is a natural generalisation, as paths can also be seen as flows with value $1$ along the path. The following reduction (\Cref{fig:mcfd_reduction}) shows that \mfdy can be regarded as a special case of \mccdy.
\begin{lemma} \label{lem:mcfd_reduction}
\mccdy is NP-hard.
\end{lemma}
\begin{proof}
Given an $s$-$t$ DAG $G=(V,E)$ and flow $X:E \to \mathbb{Y}$, we define a graph $G' = (V,E',c)$ with $E' = E \cup \{ (t,s) \}$ and $c:E' \to \{ 0, 1 \}$, $c(u,v) = 1 \iff (u,v) = (t,s)$, i.e. cost $1$ only for the edge $(t,s)$. Let $X':E' \to \mathbb{Y}, X'(e) = X(e)$ for $e \in E$ and $X'(t,s) = |X|$ be a circulation on $G'$. The cost of a Minimum Cost Circulation Decomposition $(Y'_1,w'_1),\dots,(Y'_k,w'_k)$ of $(G',X')$ is equal to $\Tilde{c}(Y'_1)+\dots+\Tilde{c}(Y'_k) = Y'_1(t,s) + \dots + Y'_k(t,s)$. Note that an $s$-$t$ path $P:E \to \{0, 1\}$ yields a circulation $Y_P:E'\to\{0,1\}$ of cost $1$, which implies $\mfdsizey{G}{X} \geq \mccdsizey{G'}{X'}$. Define flows $Y_i:E \to \mathbb{Y}$ with $Y_i(e) = Y'_i(e)$ for $i = 1,\dots,k$. Decomposing each flow $Y_i$ trivially into $|Y_i| = Y'_i(t,s)$ paths and assigning them weights $w'_i$ yields a Flow Decomposition of $(G,X)$, showing $\mccdsizey{G'}{X'} \geq \mfdsizey{G}{X}$, and thus this Flow Decomposition is minimum.
 \end{proof}

\begin{figure}
    \centering
    \includegraphics[width=0.3\textwidth]{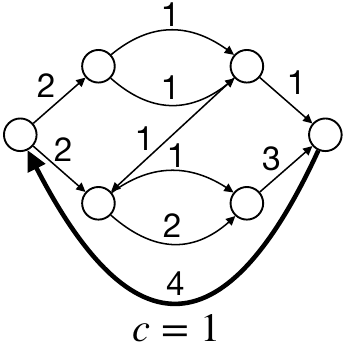}
    \caption{Reduction in \Cref{lem:mcfd_reduction} of a graph with total flow $4$ from \mfdy to \mccdy. The bold edge is the addition to the DAG and is the only edge with cost $1$ while all other edges have cost $0$.}
    \label{fig:mcfd_reduction}
\end{figure}

\begin{definition}
    Given $S \subseteq E$ and a graph $G$, we call a circulation $C$  of minimum cost satisfying $C(e)\geq1$ for all $e\in S$ Minimum Cost Circulation Cover of $S$, and we write $\mcfcS{G} = \flowcost{C}$. If $S = E$, we use $\mcfc{G}$ instead.
\end{definition}

Note that since every circulation decomposition of a graph $G$ covers the edges of non-zero circulation, $\mcfcS{G} \leq \mccdsizey{G}{X}$ with $S = \{ e \in E \mid X(e) \neq 0 \}$. This generalises the width lower bound of \mfdy to \mccdy. Given an $s$-$t$ DAG $G$, $\width{G} = \mcfc{G'}$ for the graph $G'$ obtained by the reduction in \Cref{lem:mcfd_reduction}. We will need a Minimum Cost Circulation Cover for our approximation approach:

\begin{lemma}[\cite{gabow1989faster}, Theorem 3.6]
\label{lem:min-cost-flow-cover}
Let $G = (V, E, c)$ be a graph, and $S \subseteq E$. A Minimum Cost Circulation Cover of $S$ can be computed in $\mcfcruntime$ time. 
\end{lemma}

The idea behind our approximation algorithm for \mccdint is 
that a circulation $X: E \to \mathbb{Z}$ on a graph $G$
can always be decomposed into circulations of total cost $(\lceil \log \Vert X \Vert \rceil +1) \cdot \mcfc{G}$.
We show this using two key facts: first, that $X$ can be decomposed into
$(\lceil \log \Vert X \Vert \rceil +1)$ circulations with a particular structure, and, second, that each of these circulations can be further decomposed into circulations of total cost of at most $\mcfc{G}$.
A key step in proving both these facts is a subroutine which, given an input circulation $X$, finds another circulation $Y$ with values from $\{0, \pm1\}$ only (a \emph{unitary} circulation) that matches the parity of $X$ on all edges.
Intuitively, given an input circulation $X$, such a unitary circulation $Y$
can be added to $X$ to ``fix'' its odd edges to be even, with only a small change to $X$.

\begin{lemma}\label{lem:parity-unitary}
For any circulation $X: E \to \mathbb{Z}$ on $G=(V,E,c)$, there exists a circulation $Y: E \to \mathbb{Z}$ such that $X\equiv_{2}Y$ and $\Vert Y\Vert\leq1$.
\end{lemma}

\begin{proof}

Consider the undirected graph $G_{\mathrm{odd}}=(V,E_{\mathrm{odd}})$ where $E_{\mathrm{odd}}=\{\{u,v\}\mid(u,v)\in E\text{ and }X(u,v)\text{ is odd}\}$.

Notice that every node of $G_{\mathrm{odd}}$ has even degree due to the conservation of circulation. Thus, $G_{\mathrm{odd}}$ can be written as the edge-disjoint union of cycles. Assign an arbitrary orientation to each cycle and let $E_{\mathrm{odd}}^{+}$ be the set of edges oriented in this way. Define
\[
Y(u,v)=\begin{cases}
+1 & \text{if }(u,v)\in E_{\mathrm{odd}}^{+}\\
-1 & \text{if }(v,u)\in E_{\mathrm{odd}}^{+}\\
0 & \text{if }\{u,v\}\notin E_{\mathrm{odd}}
\end{cases}
\]

Notice that $Y$ is a circulation decomposed as a sum of circulations, each along one of the edge-disjoint cycles. Moreover, $X\equiv_{2}Y$ and $\Vert Y\Vert\leq1$ by construction.
\end{proof}

Repeatedly applying \Cref{lem:parity-unitary} and dividing the resulting even circulation by 2, we obtain the the first key ingredient of the approach.


\begin{corollary}
\label{lem:pow2}
Any (non-zero) circulation $X: E \to \mathbb{Z}$ can be written as
$X=\sum_{i=0}^{\lceil\log\Vert X\Vert\rceil}2^{i}\cdot Y_{i}$,
where $Y_{i}: E \to \mathbb{Z}$ is a circulation with $\Vert Y_{i}\Vert\leq1$ for all $i$.
\end{corollary}
\begin{proof}
If $||X|| \le 1$ the result follows. 
Otherwise apply \Cref{lem:parity-unitary} to obtain $Y_0$ such that $X\equiv_{2}Y$ and $\Vert Y_0\Vert\leq1$. Since $X\equiv_{2}Y_0$, we can define $X'=(X-Y_0)/2$, and thus $X=2X'+Y_0$. Recursively repeat this procedure on $X'$ until $||X'|| \le 1$, obtaining $Y_0, \ldots, Y_k = X'$, so that $X=\sum_{i=0}^{k}2^{i}\cdot Y_{i}$. Finally, note that at each repetition, $||X||$ decreases to at most $\lceil||X||/2\rceil$, thus $k \le \lceil \log{||X||}\rceil$.
\end{proof}






The following result is the second key ingredient of our approach. It guarantees that any unitary circulation can be decomposed into two circulations of total cost of at most $\mcfc{G}$ (see \Cref{fig:unitary-to-difference} for an example).
This is by no means obvious since, among other problems, a unitary circulation may contain positive and negative values which merge and cancel each other out (as in \Cref{fig:unitary-to-difference:unitary}). 

\begin{lemma}\label{lem:unitary-to-difference}
For any circulation $X: E \to \mathbb{Z}$, $\Vert X\Vert\leq1$, there exist circulations $A,B: E \to \mathbb{Z}$ such that:
\begin{enumerate}
    \item $A,B\geq0$
    \item $X = A-B$
    \item $\flowcost{A} + \flowcost{B}\leq\mcfc{G}$
\end{enumerate}
\end{lemma}

\begin{proof}
Take $C$ such that $C\geq1$ and $\flowcost{C}=\mcfc{G}$, according to \Cref{lem:min-cost-flow-cover}. Take $D$ such that $D\equiv_{2}X+C$ and $\Vert D\Vert\leq1$, according to \Cref{lem:parity-unitary}. Also, assume $\flowcost{D}\geq0$ without loss of generality (otherwise, take $-D$, which satisfies the same properties).

Since $D\equiv_{2}X+C$, we have $C-D\pm X\equiv_{2}0$. So we can take $A=(C-D+X)/2$ and $B=(C-D-X)/2$.
\begin{enumerate}
\item Notice that $C-D\pm X\geq C-2$ since $\Vert D\Vert,\Vert X\Vert\leq1$. So, $C-D\pm X\geq-1$, since $C\geq1$. But $C-D\pm X\equiv_{2}0$ so $C-D\pm X\geq0$, whence $A,B\geq0$.
\item $A-B=\frac{C-D+X}{2}-\frac{C-D-X}{2}=X$.
\item $\flowcost{A}+\flowcost{B}=\flowcost{A+B}=\flowcost{\frac{C-D+X}{2}+\frac{C-D-X}{2}}=\flowcost{C-D}=\flowcost{C}-\flowcost{D}\leq\flowcost{C}$ since $\flowcost{D}\geq0$, and $\flowcost{C}=\mcfc{G}$. \qedhere
\end{enumerate}
\end{proof}

\begin{figure}
\centering
    \begin{subfigure}[t]{.3\linewidth}
        \centering\includegraphics[width=.9\linewidth]{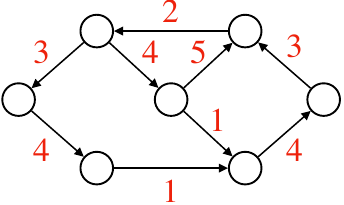}
        \caption{Edge costs}
    \end{subfigure}
    \hfill
    \begin{subfigure}[t]{.3\linewidth}
        \centering\includegraphics[width=.9\linewidth]{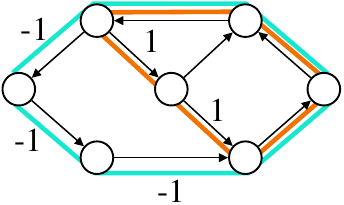}
        \caption{Unitary circulation $X$ on a graph $G$ and a decomposition into two non-negative circulations $A$ and $B$, of weight $1$ in orange (see \textbf{(e)}), and of weight $-1$ in blue (see \textbf{(f)}), respectively.}
        \label{fig:unitary-to-difference:unitary}
    \end{subfigure}
    \hfill
    \begin{subfigure}[t]{.3\linewidth}
        \centering\includegraphics[width=.9\linewidth]{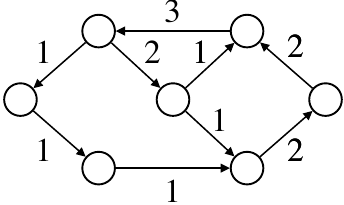}
        \caption{Circulation $C$ covering all edges of $G$, of cost $\flowcost{C} = \mcfc{G} = 42$ (\Cref{lem:min-cost-flow-cover}).}
    \end{subfigure}
\vskip\baselineskip
    \begin{subfigure}[t]{.3\linewidth}
        \centering\includegraphics[width=.9\linewidth]{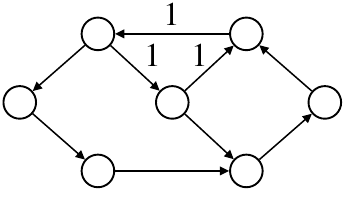}
        \caption{Unitary circulation $D$ of cost $11$ matching the parity of $X+C$, i.e., $D\equiv_{2}X+C$ (\Cref{lem:parity-unitary}).}
    \end{subfigure}
    \hfill
    \begin{subfigure}[t]{.3\linewidth}
        \centering\includegraphics[width=.9\linewidth]{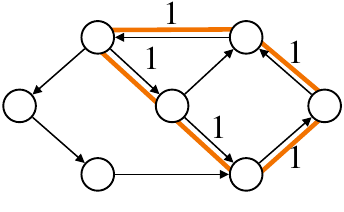}
        \caption{Circulation $A = (C-D+X)/2$ of cost $14$\label{fig:unitarytodiff_A}.}
    \end{subfigure}
    \hfill
    \begin{subfigure}[t]{.3\linewidth}
        \centering\includegraphics[width=.9\linewidth]{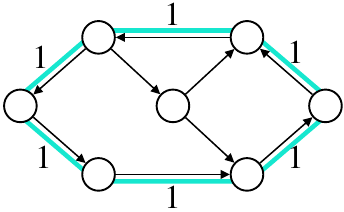}
        \caption{Circulation $B = (C-D-X)/2$ of cost $17$\label{fig:unitarytodiff_B}.}
    \end{subfigure}
    \caption{Example of \Cref{lem:unitary-to-difference} applied to a unitary circulation $X$ on a graph $G$ (for clarity, $0$ circulation values are not shown). Non-negative circulations $A$ and $B$ can be constructed so that $\flowcost{A} + \flowcost{B} \leq \mcfc{G}$ holds. We obtain a decomposition of $X$ by $X = A - B$.}
    \label{fig:unitary-to-difference}
\end{figure} 

%

Finally, expressing any circulation as a sum of at most $\lceil\log\Vert X\Vert \rceil + 1$ unitary circulations (\Cref{lem:pow2}), and decomposing each unitary circulation into two circulations with cost of at most $\mcfc{G}$ (\Cref{lem:unitary-to-difference}), we can decompose the circulation into circulations of total cost no more than $\lceil\log\Vert X\Vert\rceil + 1$ whose weights are positive and negative powers of two.


\begin{theorem}\label{thm:mdfd}
Given a graph $G=(V,E,c)$ and a circulation $X\colon E\to\mathbb{Z}$ with $k \coloneqq \lceil\log\Vert X\Vert\rceil$, there exist circulations $A_i, B_i$ for $i = 0,\dots,k$ and weights $\{w_{0},\dots,w_{k}\}\subseteq \{2^i~|~ i\in \mathbb{N}\}$, with $\flowcost{A_0 + \dots + A_k + B_0 + \dots + B_k}\leq(k+1)\cdot\mcfc{G}$ such that 
$X=w_{0}(A_{0}-B_{0})+\dots+w_{k}(A_{k}-B_{k})$.
\end{theorem}

\begin{proof}
Combine \Cref{lem:pow2} and \Cref{lem:unitary-to-difference}, getting
\begin{align*}
X & =\sum_{i=0}^{k}2^{i}\cdot Y_{i} =\sum_{i=0}^{k}2^{i}\cdot(A_i - B_i)
\end{align*}
where $\flowcost{A_i+B_i}\leq\mcfc{G}$.
\end{proof}



The proof of \Cref{thm:mdfd} suggests a straightforward algorithm for \mccdint, which we detail in \Cref{alg:mdfd-approx} and describe at a high level here. First, iteratively decompose $X$, yielding  
$ \log \lceil \Vert X\Vert \rceil + 1$ 
unitary circulations.
Then use \Cref{lem:unitary-to-difference}
to decompose each into two circulations of cost at most $\mcfc{G}$.
However, $\mcfc{G}$ is not necessarily a lower bound on \mccdint if the circulation is $0$ on some edges, and thus this approach does not directly derive an approximation. To overcome this issue, we instead find a circulation decomposition of a spanning subgraph $G'$ of $G$ for which $\mcfc{G'}$ lower bounds $\mccdsizeint{G}{X}$. Namely, we first find a minimum cost circulation cover in $G$ of the subset $S$ of edges with non-zero flow in $\mcfcruntime$ 
time (according to \Cref{lem:min-cost-flow-cover}), and then remove from $G$ any edge not covered by the circulation, obtaining $G'$.
By construction, the cost of this circulation cover is a lower bound of $\mccdsizeint{G}{X}$. Moreover, the cost of this circulation cover is exactly $\mcfc{G'}$, since every circulation cover of $G'$ is also a circulation cover of $S$ in $G$.

To prove the correctness of \Cref{alg:mdfd-approx}, we first define a a subroutine implementing \Cref{lem:parity-unitary}.


\begin{lemma}\label{lem:parity-unitary-alg}
\Cref{alg:parity-unitary} returns a unitary circulation from an input circulation $Y$ such that
$X \equiv_2 Y$, as in
\Cref{lem:parity-unitary}, in $O(m)$ time.
\end{lemma}

\begin{proof}
The correctness of the algorithm is given by \Cref{lem:parity-unitary}. Finally, the first $3$ subroutines as well as the entire \texttt{for-loop} takes $O(m)$ time.
\end{proof}



\begin{algorithm}[t]
\caption{\textsc{Unitary}(G,X): Produces a unitary circulation $Y$ from an input circulation $X$ such that $X \equiv_2 Y$,
as in \Cref{lem:parity-unitary}}
\label{alg:parity-unitary}
\begin{algorithmic}[1]
    \STATE $E_{odd} \gets$ odd edges of G, undirected
    \STATE $C \gets$ a decomposition of $G_{odd}=(V,E_{odd})$ into cycles, oriented arbitrarily
    \STATE $E_{odd}^+ \gets$ directed edges of C
    \FOR{$(u,v) \in E$}
        \IF{$(u,v) \in E_{odd}^+$}
        \STATE $Y(u,v) \gets +1$
        \ELSIF{$(v,u) \in E_{odd}^+$}
        \STATE $Y(u,v) \gets -1$
        \ELSE
        \STATE $Y(u,v) \gets 0$
        \ENDIF
    \ENDFOR
    \STATE \textbf{return} Y
\end{algorithmic}
\end{algorithm}

\begin{algorithm}[t]
\caption{Finds the circulation decomposition of \Cref{thm:mdfd}}\label{alg:mdfd-approx}
\begin{algorithmic}[1]
    \STATE Compute a minimum cost circulation cover of $\{(u,v)\in E \mid X(u,v) \neq 0\}$ \COMMENT{\Cref{lem:min-cost-flow-cover}}
    \label{line:preprocess}
    \STATE Remove from $G$ any edge not covered by this circulation cover to obtain $G'$ 
    \STATE $\mathcal{P} \gets []$, $\mathcal{W} \gets []$ \COMMENT{length-zero vectors}
    \STATE $C \gets$ circulation of cost $\mcfc{G'}$, $C \geq 1$ \label{line:mpc}
    \COMMENT{\Cref{lem:min-cost-flow-cover}}
    \STATE $D \gets \textsc{Unitary}(G',C)$; if $|D| < 0$ set $D = -D$\COMMENT{\Cref{alg:parity-unitary}}
    \label{line:D}
    \STATE $i \gets 0$
    \WHILE{$\Vert X \Vert > 1$}\label{line:startwhile}
        \STATE $Y_i \gets \textsc{Unitary}(G',X)$\COMMENT{\Cref{alg:parity-unitary}}\label{line:Y}
        \STATE $X \gets (X - Y_i)/2$
        \STATE $i \gets i + 1$
    \ENDWHILE
    \STATE $Y_{i} \gets X$
   \FOR{$j \in \{0, \ldots, i\}$ s.t. $Y_j \neq 0$}\label{line:startfor}
        \STATE $A \gets C - D + Y_j$, $B \gets C -D - Y_j$
        \STATE Concatenate $A$ and $B$ 
        to $\mathcal{P}$
        \STATE Concatenate $2^j$ and $-2^j$ to $\mathcal{W}$
   \ENDFOR
    \STATE \textbf{return} $(\mathcal{P}, \mathcal{W})$
\end{algorithmic}
\end{algorithm}


\mdfdapprox*


\begin{proof}
By \Cref{thm:mdfd} and our previous discussion, \Cref{alg:mdfd-approx} returns a circulation decomposition for $X$ with no more cost
than $(\lceil\log\Vert X\Vert\rceil)+1\cdot\mcfc{G'}\leq(\lceil\log\Vert X\Vert\rceil+1)\cdot\mccdsizeint{G}{X}$. We analyse the runtime line by line.
Lines~\ref{line:preprocess} and~\ref{line:mpc} 
take $\mcfcruntime$ 
time by \Cref{lem:min-cost-flow-cover}. 
The call to
\Cref{alg:parity-unitary} on line~\ref{line:D} takes $O(m)$ time by \Cref{lem:parity-unitary-alg},
and checking the cost of $D$ and flipping signs (if necessary) also takes $O(m)$ time.
By \Cref{lem:pow2}, the while loop on line~\ref{line:startwhile} executes at most
$\log \lceil \Vert X \Vert \rceil + 1$ times, meaning that the entire execution takes
$O(m \log \Vert X\Vert )$ time since line~\ref{line:Y} takes $O(m)$ time by
\Cref{lem:parity-unitary-alg}. Since there are at most $\log \lceil \Vert X \Vert \rceil + 1$ $Y_i$'s,
the for loop on line~\ref{line:startfor} executes at most $\log \lceil \Vert X \Vert \rceil + 1$ times.
Each execution of the for-loop finds two circulations of total cost of at most $\mcfc{G'}$ in $O(m)$ time, 
so the whole also loop takes $O(m\log\Vert X\Vert)$ time. 
Thus, the overall runtime is $O(n\log m(m+n\log n)+m\log\Vert X\Vert)$.
\end{proof}

With the reduction given in \Cref{lem:mcfd_reduction}, we obtain an approximation algorithm of the same ratio for \mfdint. However, we can improve the runtime of \Cref{lem:min-cost-flow-cover}:

\begin{corollary}
\label{lem:mcfdint-to-mfdint}
\Cref{alg:mdfd-approx} is also a $\log \lceil \Vert X \Vert \rceil + 1$-approximation for \mfdint~with runtime $O(m(n+\width{G}\log \lceil \Vert X \Vert \rceil))$.
\end{corollary}
\begin{proof}
This is directly achieved by using \Cref{thm:mdfd-approx} with \Cref{lem:mcfd_reduction} and by calculating the \width{G} according to \Cref{lem:min-flow-cover}. Note that the flows $A$ and $B$ need to be trivially decomposed into at most $\width{G}$ paths, causing the additional factor in the runtime.
\end{proof}

A theorem analogous to \Cref{thm:mdfd}
for \mccdpos~is desirable, but cannot be achieved directly with the previous methods, as \Cref{lem:parity-unitary} makes use of negative weights. However, the approach can be adapted for \mccdpos~if the input flows are width-stable (\Cref{prop:width}), and if it is possible to ``fix'' the odd flows
to be even with only $\mcfc{G}$ unitary flows, which we leave as an open question.




\section{Solving the $k$-Flow Weight Assignment Problem}\label{sec:k-flow-weight-assignment}


In this section, we consider a restriction of MFD from~\cite{kloster2018practical} (see \Cref{figure:counter-example} for an example).
\begin{definition}[$k$-Flow Weight Assignment]
Given a flow $X: E \to \mathbb{Y}$ on a graph $G = (V, E)$ and a set of $s$-$t$ paths $\{P_1,\dots,P_k\}$, the problem of finding an assignment of weights to the paths, such that they form a flow decomposition of $(G,X)$, is called \emph{$k$-Flow Weight Assignment} ($k$-FWA). We write $k$-FWA$_{\mathbb{Y}}$ if we require the path weights to belong to $\mathbb{Y}$.
\end{definition}

Given $k$ $s$-$t$ paths, $k$-FWA can be solved by a linear system defined by $Lw=X$, where $X_j \in \mathbb{Y}$ is equal to the flow $X(e_j)$ of the edge $e_j$ (we identify flows $X:E\to\mathbb{Y}$ with vectors $X\in\mathbb{Y}^m$) and $L$ is the $m\times k$ $0/1$ matrix with $L_{i,j} = 1$ if and only if path $P_j$ crosses edge $e_i$.
The resulting solution $w \in \mathbb{Y}^{k}$ is the weight assignment to each path. For a flow graph $(G,X)$, we denote by $L_\mathbb{Y} = L_\mathbb{Y}(P_1, \dots, P_k) = \{ w \in \mathbb{Y}^k | X = \sum_{j = 1}^k P_k w_k \}$ the linear system corresponding to the paths $P_1, \dots, P_k$.

We shortly discuss how to solve \kfwaint.
The linear system defined by the paths is a system of linear Diophantine equations.
It is well known that integer solutions to such systems can be found in polynomial time; see, e.g.,~\cite[Chapter 5]{Shrij1986linearprog}. 

Solving \kfwapos turns out to be more difficult, 
its the linear system contains the inequality $w \ge 0$. In fact, it was shown~\cite{kloster2018practical} that \kfwapos is NP-hard.
The program Toboggan~\mbox{\cite{kloster2018practical}} implements a linear FPT algorithm for \mfdpos. and one step of the algorithm is to solve \kfwapos using an ILP~\cite{kloster2018practical}.
The authors state the following conjecture.

\begin{conjecture}[\cite{kloster2018practical}]
\label{conj:kloster}
If $(P_1, \dots, P_k)$ are the paths of a minimum flow decomposition of $(G,X)$, then the linear system $L_\mathbb{N}(P_1,\dots,P_k)$ has full rank $k$.
\end{conjecture}

In case of a fractional decomposition (in which the weights of the paths are allowed to be rational non-negative numbers), it is indeed true that the induced linear system is of full rank $k$~\cite{vatinlen2008simple}. As mentioned in the introduction, if the conjecture turned out to be true for natural numbers, Toboggan could avoid resorting to solving an ILP, 
since just solving the standard linear system at hand would return its unique solution. 
As observed by the authors, this would decrease the asymptotic worst case upper bound of Toboggan.

\begin{figure}
    \includegraphics[width=1.\textwidth]{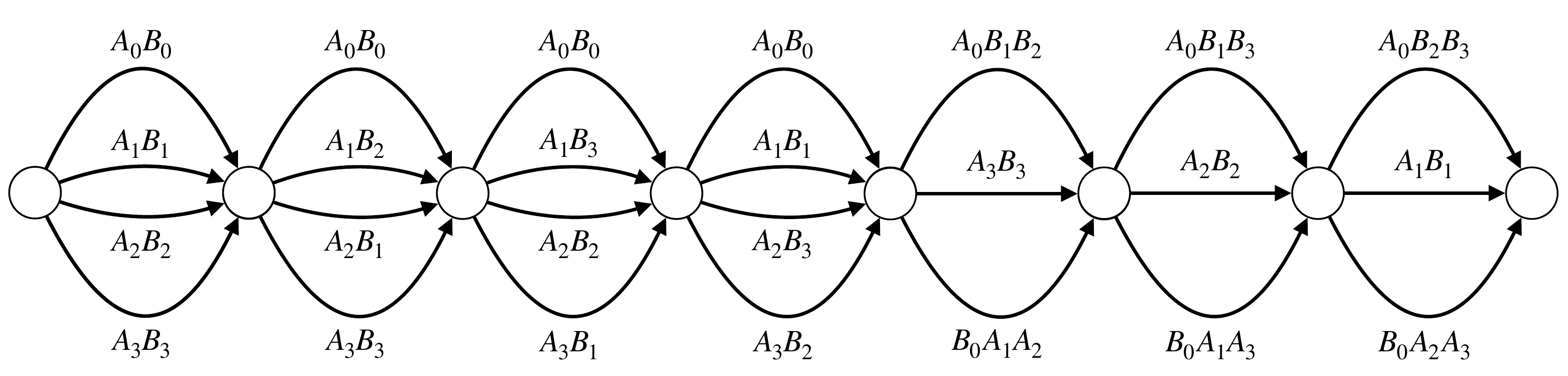}
    \caption{
    Paths $A_i$ and $B_i$ ($i \in \{0,1,2,3\}$), each edge being labeled with the paths it appears in. Assign to each path $A_i$ weight $a_i$, and to each path $B_i$ weight $b_i$, such that $a_0=b_0=3$, and $a_i = 6^{2i}+1$ 
    and $b_i = 6^{2i+1}+5$ for $i=1,2,3$. Define the flow $X$ on $G$ as $X = \sum_{i=0}^{3}a_iA_i + \sum_{i=0}^{3} b_iB_i$. Note that these weights are a solution of \kfwapos on input $(G,X)$ with given paths $A_i$, $B_i$ ($i \in \{0,1,2,3\}$).
    }
    \label{figure:counter-example}
\end{figure}



We show that this conjecture is false using a counterexample. Consider the input for \kfwapos from \Cref{figure:counter-example} and the solution therein. We now give another solution for \kfwapos on this input, namely the following path weights: $a_0=5, b_0 =1$, and $a_i = 6^{2i} + 2, b_i = 6^{2i+1} + 4$, for $i = 1, 2, 3$.
One can easily verify that this is another solution to \kfwapos on the input in \Cref{figure:counter-example}, thus proving  
that the rank of the corresponding linear system is strictly less than $8$.

To disprove \Cref{conj:kloster}, it remains to show that any flow decomposition contains at least $8$ paths. Due to the technicality of this proof (and its exhaustive case-by-case analysis), we only explain the intuition behind the construction in \Cref{figure:counter-example} and behind the correctness proof. However, as an additional check we also ran both Toboggan~\cite{kloster2018practical} and a recently developed ILP solver for \mfdpos~\cite{Dias:2022uv} on this instance, both returning $\mfdsizepos{G}{X} = 8$.


The intuition is as follows. The graph can be divided into two parts: the graph induced by the first $5$ vertices in topological order (left part) and the one induced by the last $4$ (right part). We say that a path is \textbf{fixed} if every minimum flow decomposition of the graph contains this path.
The paths $A_i$ and $B_i$ have expoentially growing weight for growing $i$ and get shuffled around with different permutations of the paired labels $A_iB_j$ on the left part. Due to the exponential growth, ensuring the correct parity on all edges of the right part, we can fix the paths $A_i$ and $B_i$ for $i = 1,2,3$.
This allows us to interpret flow decompositions of less than $8$ paths as decompositions with $8$ paths, where either $A_0$ or $B_0$ carries weight $0$. Consider a flow decomposition where we assign two paths of weights $\lambda_1$ and $\lambda_2$ on the edges labeled $A_0B_0$. For any $\delta \geq 0$, $(\lambda_1 - \delta) + (\lambda_2 + \delta) = a_0 + b_0$ and equivalently for all other edges on the left part. If we decrease $\lambda_1$ by some $\delta > 0$, the weights of $B_1$, $B_2$ and $B_3$ each increase by $\delta / 2$. And thus, $\delta$ must be even. Due to the parity of $a_0$ and $b_0$, they can never reach $0$.

\section{Conclusions}
In this paper we have shown for the first time that width, a natural lower bound for MFD, is also useful when investigating its approximability. On the one hand, using width is a key insight in understanding where greedy path removal heuristics fail. On the other hand, graphs where width is well-behaved (e.g.,~series-parallel graphs) have a guaranteed approximation factor. Moreover, we generalised MFD to the problem to minimising the cost of a circluation decompisition, and showed that the integer version can be approximated even better by combining parity arguments of unitary circulations and a decomposition of such circulations of cost equal to the minimum cost to cover the graph. Finally, we have corroborated the complexity gap between the positive integer and the full integer case by disproving a conjecture from~\cite{kloster2018practical} (also motivating the heuristic in~\cite{shao2017theory}), which would have had sped up their FPT algorithm for \mfdpos.

Our results open up new avenues for further research on MFD. For example, can the width help find larger classes of graphs for which some greedy path removal (or even some sort of greedy \emph{path cover} removal) algorithms have a guaranteed approximation factor? 
Can we get $\Omega(n)$ worst case approximation ratio of greedy-weight for dense graphs without parallel edges?
Can the power-of-two decomposition approach be applied with other factors besides two? Can better path cover-like lower bounds help (e.g., path covers which cannot use an edge more times than its flow value, also computable in polynomial time)? 
How do our algorithms perform in practice?

\section{Acknowledgements}
This work was partially funded by the European Research Council (ERC) under the European Union's Horizon 2020 research and innovation programme (grant agreement No.~851093, SAFEBIO), partially by the Academy of Finland (grants No.~322595, 328877, 352821, 346968), and partially by the National Science Foundation (NSF) (grants No.~1661530,1759522).

\bibliography{main}

\end{document}